\documentclass[1p]{elsarticle} %

\usepackage{graphicx}
\usepackage{mathptmx}

\usepackage{amsmath}
\usepackage{amsthm}
\usepackage{amssymb}

\newtheorem{theorem}{Theorem}
\newtheorem{proposition}{Proposition}
\newtheorem{corollary}{Corollary}

\newtheorem{example}{Example}
\newtheorem{definition}{Definition}

\def\d{\mathop{\mathrm d\null}\!}
\def\GL{\mathrm{GL}}
\def\SL{\mathrm{SL}}
\def\e{\mathrm e}
\def\RR{\mathbb R}
\def\eqref#1{{\rm(\ref{#1})}}

\begin{document}

\title{Matching van Stockum dust to Papapetrou vacuum\tnoteref{ded}}
\tnotetext[ded]{In commemoration of prof.~Jan Horsk\'y (1940--2012).}


\author{Michal Marvan\fnref{orcid}}
\address{Mathematical Institute in Opava \\ 
Silesian University in Opava \\
Na Rybn\'{\i}\v{c}ku 626/1 \\
746\,01 Opava, Czech Republic} 
\ead{Michal.Marvan@math.slu.cz}
\fntext[orcid]{\tt ORCID 0000-0001-8685-1313.}


\begin{abstract}
Addressing a long-standing problem,
we show that every van Stockum dust can be matched to
a 1-parametric family of non-static Papapetrou vacuum metrics,
and the converse.
The boundary, if existing, is determined 
by vanishing of certain first-order invariant on the vacuum side. 
Moreover, we establish a relation to Ehlers and Kramer--Neugebauer 
transformations, which allows us to look for dust clouds with a 
prescribed boundary.
Explicit examples include the   
Bonnor metric and a new vacuum exterior to the 
Lanczos--van Stockum dust metric, as well as dust clouds
with nontrivial topology. 
\end{abstract}


\maketitle


\section{Introduction}

General relativity is an inexhaustible source of mathematical 
challenges.
Although exact vacuum and dust metrics are known in abundance,
much less is known about possible dust-vacuum configurations.
The central result of this paper is a correspondence 
between the van Stockum~\cite{WJvS-1937} class of dust metrics 
and the Papapetrou~\cite{AP-1953} class of vacuum metrics 
such that the corresponding metrics match,
i.e., can be glued together along a common boundary.  
As an auxiliary result, we prove that the boundary, if existing, 
is determined by vanishing of certain first-order 
invariant on the vacuum side (Corollary~\ref{cor:boundary}).

We also relate the matching van Stockum--Papapetrou pairs to static 
Weyl vacuum metrics and, in particular, to axisymmetric harmonic 
functions in~$\RR^3$. This yields two equivalent ways to obtain 
dust clouds with a prescribed boundary.
Explicit worked-out examples include a vacuum exterior to the 
Bonnor~\cite{WBB-1977} dust metric,
a dust source for the
Halilsoy~\cite{MH-1992} vacuum metric,
a new cylindrically non-symmetric vacuum exterior to the 
Lanczos~\cite{KL-1924}--van Stockum~\cite{WJvS-1937} dust 
metric, as well as two toroidal dust-vacuum configurations.

The problem of matching a general van Stockum dust to vacuum has 
been open for decades (Bonnor \cite[\S~(3)]{WBB-1982},
Viaggiu~\cite{SV-2007}, 
Zingg et al.~\cite{Z-A-T-2007},
G\"urlebeck~\cite{NG-2009}).
The list of previously known van Stockum--vacuum $C^1$-metrics 
consists of only two items:
the van Stockum's rigidly rotating dust cylinder~\cite{WJvS-1937} 
composed of the Lanczos dust~\cite{KL-1924} and
a Lewis vacuum~\cite{TL-1932} (three Killing vectors);
and the Zsigrai~\cite{JZ-2003} metric composed of the 
Luk\'acs--Newman--Sparling--Winicour~\cite{L-N-S-W-1983} dust metric
and the vacuum NUT metric~\cite{N-T-U-1963} (four Killing vectors). 
 
We only consider glued metrics of class $C^1$ on the boundary 
in the sense of Lichnerowicz~\cite{AL-1955},
which we call simply $C^1$-metrics.
Dust-vacuum $C^1$-metrics model massive dust clouds, 
leaving the well-known known $C^0$-continuous relativistic analogues of 
infinitely thin rotating Newtonian discs outside the scope of this paper.

\section{Matching}
\label{sect:gluing}

To start with, we present Lichnerowicz's matching as a particular
instance of the contact condition. The latter allows for 
a simultaneous treatment of a number of derivatives.

Let $M$ be a space-time manifold, partitioned by a
two-sided hypersurface $B$.
By definition, functions $f^{\rm (I)}, f^{\rm (II)} \in C^\infty M$ 
satisfy the condition $f^{\rm (I)} \equiv_{B}^{k} f^{\rm (II)}$ of 
$k$th-order {\it contact} along $B$ if and only if, 
in each coordinate patch,
\[f^{\rm (I)}_{,i_1 \cdots i_l}|_B
 = f^{\rm (II)}_{,i_1 \cdots i_l}|_B 
\quad \text{for all $0 \le l \le k$},
\]
where $f_{,i_1 \cdots i_l}$ denote the partial derivatives and
$f|_B $ is the restriction to $B$. 

The relation $\equiv_{B}^{k}$ is an equivalence 
relation and, what is of utmost importance, 
a congruence of the algebra of functions 
on~$M$ with arbitrary $C^k$-continuous multivariate functions as 
operations. This can be formulated as the following proposition,  
which can be easily proved by using the chain rule.

\begin{proposition}
\label{prop:cong}
Assuming $f_1^{\rm (I)} \equiv_{B}^{k} f_1^{\rm (II)}$, 
$\dots$, $f_m^{\rm (I)} \equiv_{B}^{k} f_m^{\rm (II)}$,
let $F(f_1,\dots,f_m)$ be a $C^k$-continuous function in 
a neighbourhood of the image 
$f_1^{\rm (I)} B \times \cdots  \times f_m^{\rm (I)} B
 = f_1^{\rm (II)} B \times \cdots  \times f_m^{\rm (II)} B$. 
Then
\begin{equation}
\label{eq:cong}
F(f_1^{\rm (I)},\dots,f_m^{\rm (I)}) 
  \equiv_{B}^{k} F(f_1^{\rm (II)},\dots,f_m^{\rm (II)})
\end{equation}
holds.
\end{proposition}

Assume $M$ multi-partitioned into $n$ connected manifolds
$\bar U^{(i)} = U^{(i)} \cup \partial U^{(i)}$ with pairwise disjoint
boundaries $\partial U^{(i)}$, separated by a two-sided (discontinuous)
hypersurface 
$$
B = \bigcup_{i = 1}^n \partial U^{(i)}.
$$ 
Let $f^{\rm (I)}, f^{\rm (II)} \in C^\infty M$,
assuming \smash{$f^{\rm (I)} \equiv_{B}^{k} f^{\rm (II)}$}. 
Let $\nu_i = \rm I$ or $\rm II$ at will, $i = 1,\dots,n$
(we are free to keep or swap).
Then the unique function $f$ on $M$ that coincides with
$f^{\rm (\nu_i)}$ on $\bar U^{(i)}$ is $C^k$-continuous.
This way of obtaining a $C^k$-continuous function $f$ is known as 
gluing along the boundary~$B$; there are $2^n$ possible combinations
corresponding to $2^n$ possible choices of $\{\nu_i\}_{i  = 1,\dots,n}$.
Generalisation to three and more functions 
$f^{\rm (I)}, f^{\rm (II)}, f^{\rm (III)}, \ldots \in C^\infty M$
is straightforward.

Continuity, contact and gluing extend to tensors in an obvious way.
A metric $\mathbf g$ is said to be of class $C^k$, or a $C^k$-metric, 
if there exist coordinates $x^\mu$ with respect to which the metric 
coefficients $\mathbf g_{\mu\nu}$ are of class $C^k$.
The condition of $C^k$-continuity of a metric is 
preserved under all coordinate transformations of class~$C^{k+1}$.
Similarly for other tensors.

Given two $C^\infty$-metrics 
$\mathbf g^{\rm (I)}, \mathbf g^{\rm (II)}$
determined by their coordinate components 
$\mathbf g_{\mu\nu}^{\rm (I)}, \mathbf g_{\mu\nu}^{\rm (II)}$ with
respect to coordinates $x^\mu$, then 
\begin{equation}
\label{eq:matching}
\mathbf g_{\mu\nu}^{\rm (I)} \equiv_{B}^{k} \mathbf g_{\mu\nu}^{\rm (II)}
\end{equation}
implies that the ``piecewise'' metric obtained by gluing the coordinate 
components along $B$ is $C^k$-continuous.
Again, condition~\eqref{eq:matching} is preserved under coordinate 
transformations of class~$C^{k+1}$.
As with functions, $2^n$ combinations corresponding to $n$ 
connected components are possible.

If $k = 1$, conditions~\eqref{eq:matching} 
coincide with the Lichnerowicz matching conditions 
(``conditions de raccordement''~\cite[p.~61]{AL-1955}) of general relativity.
The coordinates $x^\mu$ with respect to which the matching 
conditions~\eqref{eq:matching} hold are called admissible. 
Perhaps more appropriately, we shall call them {\it shared} coordinates.

The $C^1$-metrics occur as solutions of the 
Einstein equations with a discontinuous energy-momentum tensor, 
of which dust clouds are important examples.
Concerning the dust movement,
it is required that the trajectories (actually geodesics) 
the dust particles move along are either disjoint with the 
joining hypersurface or contained in it~\cite[p.~63]{AL-1955},
which is obviously true in our setting.

\section{The choice of coordinates}

Metrics we intend to match should be written in convenient shared 
coordinates.
We do not insist the boundary $B$ 
to have any special coordinate expression. 
Actually, our results would be quite difficult to obtain under
any such stipulation.

Both the van Stockum and the Papapetrou metrics admit 
two commuting and orthogonally transitive Killing vectors, 
one time-like and one space-like. 
Assuming shared Killing vectors, 
the glued metric can be written in the form \cite{MP-S-1993} 
\begin{equation}
\label{eq:gg}
\mathbf{g} = \tilde g_{ij}(t^1,t^2)\d t^i\d t^j
 + h_{kl}(t^1,t^2)\d z^k\d z^l,
\end{equation}
where 
$i,j,k,l = 1,2$.
The Killing vectors are \smash{$\xi_{(k)} = \partial/\partial z^k$} and
also linear combinations thereof.
The term $\tilde g_{ij}\d t^i\d t^j$ represents the 
metric on the quotient space of Killing orbits~\cite{RG-1971}, 
referred to as the {\it orbit space}, while 
$h_{kl} = \mathbf{g}(\xi_{(k)},\xi_{(l)})$. 
In what follows, $\det h < 0$ everywhere (one of the Killing vectors
is time-like) and, consequently, $\det g > 0$ (the orbit metric
is Riemannian).
We also assume that the boundary hypersurface retains the symmetry 
and projects to a curve in the two-dimensional orbit space.
The matching problem reduces to finding that curve.

Local coordinate transformations 
\begin{equation}
\label{eq:transf:xyz}
\bar{t}^{i} = \Phi^{i}(t^{1},t^{2}),\qquad
\bar{z}^{k} = A^{k}_{l}\,z^{l}
\end{equation}
preserve the form~\eqref{eq:gg} of the metric.
Here $\Phi^{i}(t^{1},t^{2})$ are local coordinate transformations
in the orbit space, while
$A = (A^m_n)\in \mathrm{GL}_2$ are constant matrices
acting on the components $h_{kl}$ by
\begin{equation}
\label{eq:transf:GL2}
\bar h_{kl} = A_k^m h_{mn} A^n_l,
\end{equation}
(linear transformation of the Killing vectors).

Furthermore, the glued orbit space admits 
isothermal coordinates,
i.e., coordinates $x,y$ in which the orbit metric assumes the 
form $p(x,y)\,(\d x^2 + \d y^2)$.
This follows from the Korn--Lichtenstein 
theorem~\cite{AKorn-1914,LL-1911} 
(see also~\cite[p.~262]{H-W-1953} or \cite[p.~772]{SSC-1955}), 
according to which 
local isothermal coordinates exist under the assumption of the 
H\"older $C^{0,\alpha}$-continuity of degree $0 < \alpha \le 1$. 
Since the $C^1$-continuity required by the Lichnerowicz 
conditions implies the H\"older $C^{0,\alpha}$-continuity, 
shared isothermal coordinates $x,y$ are guaranteed to exist on 
the orbit space.
This means that the glued space-time metric can be written 
in the Lewis--Papapetrou~\cite{TL-1932} form 
\begin{equation}
\label{eq:Lewis:Papa}
\mathbf g = \e^p\,(\d x^2 + \d y^2)
 + h_{kl}\d z^k\d z^l.
\end{equation}
The remaining freedom to transform $z^1,z^2$ is 
represented by the $\GL_2$-action~\eqref{eq:transf:GL2}. 
Thus, two Lewis--Papapetrou metrics 
$\mathbf g^{\rm (I)} = \e^{p^{\rm (I)}}(\d x^2 + \d y^2)
 + h^{(\mathrm I)}_{kl}
   \d z^{k} \d z^{l}$, 
$\mathbf g^{\rm (II)} = \e^{p^{\rm (II)}}(\d x^2 + \d y^2)
 + h^{(\mathrm{II})}_{kl}
   \d z^{k} \d z^{l}
$
match if and only if
\begin{equation}
\label{eq:matching:LP}
p^{(\rm I)} \equiv_{B}^{1} p^{(\rm II)}, \quad
h^{(\rm I)}_{kl} \equiv_{B}^{1} A_k^m h^{(\rm II)}_{mn} A^n_l. 
\end{equation}
Here $A \in \GL_2$ is an arbitrary constant matrix.

\section{A useful invariant}

According to~\cite{MM-OS-2008},  
metrics~\eqref{eq:gg} possess four algebraically independent 
first-order scalar invariants $C_\rho,C_\chi,Q_\chi,Q_\gamma$
with respect to coordinate 
transformations~\eqref{eq:transf:xyz}. 
Given a Lewis--Papapetrou metric~\eqref{eq:Lewis:Papa},
the scalar invariant $Q_\chi(\mathbf g)$ is
\[
\begin{aligned}
&Q_\chi(\mathbf g) = \frac{\det\chi}{\e^{2p}}, 
\quad
\chi = \frac{1}{\det h}
\left|\begin{array}{@{}cc@{}}
\d h_{11} & \d h_{12} \\ \d h_{12} & \d h_{22}
\end{array}\right|,
\\
&\chi_{ij} = \frac{1}{2 \det h}
\left(
\left|\begin{array}{@{}cc@{}}
h_{11,i} & h_{12,j} \\ h_{12,i} & h_{22,j}
\end{array}\right| 
 + 
\left|\begin{array}{@{}cc@{}}
h_{11,j} & h_{12,i} \\ h_{12,j} & h_{22,i}
\end{array}\right|
\right).
\end{aligned}\]
Being a first-order invariant, $Q_\chi(\mathbf g)$ is continuous
for every $C^1$-metric $\mathbf g$.
Therefore, the matching conditions~\eqref{eq:matching} imply
\begin{equation}
\label{eq:matching:Q}
Q_\chi(\mathbf g^{(\rm I)})|_{B}
 = Q_\chi(\mathbf g^{(\rm II)})|_{B}. 
\end{equation}
Condition~\eqref{eq:matching:Q} provides a tool to locate the 
boundary $B$, while avoiding the use of the unknown matrix $A$.
An analogous condition can be written for every
first-order scalar invariant.
However, $Q_\chi$ is special in that it vanishes for
dust metrics of interest in this paper, see
Proposition~\ref{prop:Qchi:circular}.


A dust metric of the Lewis--Papapetrou class is said 
to be {\it circular} if the 4-velocity $\mathbf U$ 
belongs to the distribution spanned by the Killing vectors,
i.e., can be written as
$\mathbf U = U^k \xi_{(k)} = U^k\,\partial/\partial z^k$,
$k = 1,2$.

\begin{proposition}
\label{prop:Qchi:circular}
Every circular Lewis--Papapetrou dust metric satisfies
\[
Q_\chi(\mathbf g) = 0.
\]
\end{proposition}

\begin{proof}
The energy-momentum tensor is
$\mathbf T^{\mu\nu} = \rho^{\rm dust} \mathbf U^\mu \mathbf U^\nu$,
where $\rho^{\rm dust}$ is the density.
The vanishing of the divergence, $\mathbf T^{\mu\nu}{}_{;\nu} = 0$, 
implies 
$h_{kl,x} U^k U^l = 0 = h_{kl,y} U^k U^l$,
which is a system of homogeneous equations on $U^1,U^2$.
This system has a nonzero solution if and only if its resultant 
is zero, which is equivalent to $Q_\chi(\mathbf g) = 0$.
\end{proof}

It easily follows that the boundary between vacuum and circular dust, 
if existing, is almost uniquely determined on the vacuum side. 
By almost uniquely we mean that for each boundary separating two
connected components we can choose to perform gluing or not,
see Sect.~\ref{sect:gluing}.

\begin{corollary}
\label{cor:boundary}
If a Lewis--Papapetrou metric $\mathbf g$ matches a van Stockum 
dust metric along a boundary $B$, 
then $Q_\chi(\mathbf g)|_B = 0$. 
\end{corollary}

\section{Van Stockum metrics}
Henceforth we restrict attention to dust metrics of the 
van Stockum class~\cite{WJvS-1937}.
Van Stockum's dust is isometrically flowing~\cite{GW-1968}, 
meaning that the 4-velocity $\mathbf U^a$ is a Killing vector.  

We briefly recall the derivation of van Stockum equations
\cite[Eqs.~(5.5)--(5.10)]{WJvS-1937}. 
Renaming the coordinates as
$t^1 = x$, $t^2 = y$, $z^1 = \phi$, $z^2 = t$, where
$x,y$ are isothermal 
(see Section~\ref{sect:gluing}),
coordinates $\phi,t$ can be chosen as comoving, 
i.e., in such a way 
that the 4-velocity $\mathbf U^a$ equals $\partial_t$.
Then the coefficient at $\d t^2$, which is the magnitude 
$\mathbf g_{\mu\nu} \mathbf U^\mu \mathbf U^\nu$
of the 4-velocity, equals $-1$.
Consequently, a general van Stockum dust metric can be written 
in the form
\begin{equation}
\label{eq:dust:ds}
\mathbf g^{\mathrm d} = \e^p\,(\d x^2 + \d y^2)
 + r^2\d \phi^2
 - (f\d \phi + \d t)^2.
\end{equation}

The dust moving circularly, the Einstein equations 
$\mathbf R_{\mu\nu} - \frac12 \mathbf R \mathbf g_{\mu\nu}
 = \mathbf T^{\mu\nu}$
imply
$\mathbf R^3{}_3 + \mathbf R^4{}_4 = 0$. 
This gives $r_{xx} + r_{yy} = 0$. 
Then $r$ is a harmonic function that, if non-constant, can serve 
as one of the isothermal coordinate functions
(Weyl's canonical coordinates~\cite[p.~137]{HW-1917} 
or~\cite{TL-1932}). 
With $r = x$, the Einstein equations reduce to
\begin{equation}
\label{eq:dust:fp}
f_{xx} + f_{yy} - \frac{f_x}{x} = 0,
\quad
p_x = \frac{f_y^2 - f_x^2}{2 x},
\quad
p_y = -\frac{f_x f_y}{x},
\end{equation}
while the dust density turns out to be
\begin{equation}
\label{eq:dust:rho}
\rho^{\rm d} = 
\frac{f_x^2 + f_y^2}{x^2 \e^p}.
\end{equation}
The case of constant $r$ 
(or, invariantly, $C_\rho = 0$ by~\cite{DCF-MM-2020})
has been studied by 
Hoenselaers and Vishveshwara~\cite{H-V-1979},
who also found a matching vacuum metric.

\section{Papapetrou metrics}
\label{sect:Papa}


The vacuum part is chosen to be the Papapetrou metric~\cite{AP-1953}.
The necessary details are set forth below with the aim to expose 
similarity to the dust case.
In isothermal coordinates,
\begin{equation}
\label{eq:papa:ds}
\mathbf{g^{\mathrm v}} = \e^q\,(\d x^2 + \d y^2) + \frac{r^2}{v} \d \phi^2
 - v (w\d \phi + \d t)^2
\end{equation}
by appropriate choice of the field variables in eq.~\eqref{eq:gg}.
Again, we have $\mathbf R^3{}_3 + \mathbf R^4{}_4 = 0$,
giving $r_{xx} + r_{yy} = 0$, i.e., $r$ is a harmonic function. 
Since $\det h = -x^2$ on the dust side
and $\det h = -r^2$ on the vacuum side are required to have
a first-order contact on the boundary, we have simultaneous 
Dirichlet and Cauchy boundary conditions $r = \pm x$,
$r_x = \pm 1$, $r_y = 0$ and we are left with $r = x$ everywhere.
This also restricts the choice of the matrix $A$ in the matching 
condition~\eqref{eq:matching:LP} to $\SL_2$.
 
The vacuum Einstein equations reduce to
\begin{equation}
\begin{split}
\label{eq:Einst:vac}
& w_{xx} + w_{yy} - \frac{w_x}{x} = -2 \frac{v_x w_x + v_y w_y}{v},
\\
& v_{xx} + v_{yy} + \frac{v_x}{x} = \frac{v_x^2 + v_y^2}{v}
 - v^3 \frac{w_x^2 + w_y^2}{x^2},
\\
& q_x = -\frac{v_x}{v} + \frac{x}{2 v^2} (v_x^2 - v_y^2)
 - \frac{v^2}{2 x} (w_x^2 - w_y^2),
 \\
& q_y =  -\frac{v_y}{v} + \frac{x v_x v_y}{v^2}
 - \frac{v^2 w_x w_y}{x}.
\end{split}
\end{equation}
The Papapetrou class of metrics is determined by the condition
\begin{equation}
\label{eq:Papa:cond}
v_x w_x + v_y w_y = 0
\end{equation}
\cite[eq.~(3.1)]{AP-1953}.
If $w = $ const, then the Riemann tensor becomes zero and the metric is 
flat.
Otherwise $w_x \ne 0$ or $w_y \ne 0$.
We work out the case $w_y \ne 0$ (the other one leads to
the same result).
Denoting
\begin{equation}
\label{eq:Papa:c}
c^2 = \frac{x^2 v_x^2}{v^2 w_y^2} + v^2 > 0,
\end{equation}
it is easily checked that $c_x = c_y = 0$ in consequence of  
equations~\eqref{eq:Einst:vac} and~\eqref{eq:Papa:cond}.
Hence, expression~\eqref{eq:Papa:c} is a first integral,
meaning that solutions are classified by $c > 0$. 
However, system~\eqref{eq:Einst:vac}, the
Papapetrou condition~\eqref{eq:Papa:cond} 
and formula~\eqref{eq:Papa:c} are preserved under a 
three-dimensional Lie group of coordinate 
transformations~\eqref{eq:transf:xyz}, 
one of the generators being
\[
\begin{array}{r|cccc}
             & w      & v         & c         & q 
\\\hline
\mathcal S_a & \e^a w & \e^{-a} v & \e^{-a} c & q 
\end{array}.
\]
Using transformation $\mathcal S_a$, one can always normalise 
$c$ to $1$, which we assume henceforth.
Setting 
\begin{equation}
\label{eq:Papa:qsuv}
v = 1/{\cosh u}, 
\quad
\e^q = \frac{\e^s}{v} = \e^s \cosh u,
\end{equation}
the whole system~\eqref{eq:Einst:vac},~\eqref{eq:Papa:cond} 
simplifies to 
\begin{equation}
\label{eq:Papa:wus}
\begin{split}
& w_{xx} + w_{yy} - \frac{w_x}{x} = 0, 
\\ &
u_x = -\frac{w_y}{x},
\quad
u_y = \frac{w_x}{x},
\quad
s_x = -\frac{w_x^2 - w_y^2}{ 2x},
\quad
s_y = -\frac{w_x w_y}{x}.
\end{split}
\end{equation}

\section{Locating the boundary}


According to Corollary~\ref{cor:boundary}, the boundary $B$
must satisfy $Q_\chi(\mathbf g^{\mathrm v})|_B = 0$.
Computing the invariant $Q_\chi$ of a general Papapetrou metric 
$\mathbf g^{\mathrm v}$ given by eq.~\eqref{eq:papa:ds} under 
identification~\eqref{eq:Papa:qsuv},
we obtain  
$$
Q_\chi(\mathbf g^{\mathrm v}) = \frac{\e^s \sinh u}{r^4 \cosh^4 u}
P(\sinh u, \cosh u, w_x, w_y), 
$$
where $P$ is a polynomial.
Consequently, the boundary is either 
$u = 0$ or  $P = 0$.
We continue with $u = 0$, which leads to the general result
presented below.


\begin{definition} \rm
A dust metric $\mathbf g^{\rm d}$ determined by field variables 
$f,p$ satisfying equations~\eqref{eq:dust:fp} and a normalised 
non-flat Papapetrou metric $\mathbf g^{\rm v}$ determined by field 
variables $w = f$, $s = p$ and $u$ satisfying 
equations~\eqref{eq:Papa:wus} are called {\it companions}.
\end{definition}

Under relations~\eqref{eq:Papa:qsuv}, 
the companion dust and vacuum metrics can be written as
\begin{equation}
\label{eq:dust:ds:E}
\mathbf g^{\rm d} = \e^p\,(\d x^2 + \d y^2)
 + x^2\d \phi^2
 - (f\d \phi + \d t)^2
\end{equation}
and
\begin{equation}
\label{eq:papa:ds:NK}
\mathbf{g}^{\rm v} = \e^{p}\cosh u\,(\d x^2 + \d y^2) 
 + x^2 \cosh u \d \phi^2
 - \frac{(f\d \phi + \d t)^2}{\cosh u},
\end{equation}
respectively.

The companion correspondence between the dust and the vacuum metrics
is one to continuum,
since system~\eqref{eq:Papa:wus} determines $u$ uniquely up to an 
integration constant.

\begin{theorem}
\label{thm:main}
The companion dust and vacuum metrics match along the boundary 
located at $u = 0$.
\end{theorem}

\begin{proof} 
Obviously from Proposition~\ref{prop:cong} and 
formulas~\eqref{eq:dust:ds:E} and~\eqref{eq:papa:ds:NK},
\[
\mathbf g^{\rm d}_{ij} \equiv_{\{u = 0\}} \mathbf g^{\rm v}_{ij},
\]
since $\cosh u \equiv_{\{u = 0\}} 1$ (omitting the superscript~1).
\end{proof}

A $C^1$-metric obtained by gluing according to 
Theorem~\ref{thm:main} will be called the
{\it van Stockum--Papapetrou $C^1$-metric} or the
{\it van Stockum--Papapetrou dust cloud}.


The normalisation $c = 1$ we made in Section~\ref{sect:Papa}
ensures that the coefficients 
$h^{\rm d}_{kl} \equiv h^{\rm v}_{kl}$ match
without the need for an adjustment by 
transformation~\eqref{eq:transf:GL2}.

\section{Electrostatic analogy}
\label{sect:el}

Dust clouds of a prescribed shape can be obtained 
in terms of the classical potential theory. 
Recall that
companion metrics are determined by system~\eqref{eq:Papa:wus}.
Eliminating $w$, one obtains the
equivalent system
\begin{equation}
\label{eq:Papa:uws}
\begin{aligned}
& u_{xx} + u_{yy} + \frac{u_x}{x} = 0,
\\
& w_x = x u_y,
\quad
w_y = -x u_x,
\quad
s_x = x \frac{u_x^2 - u_y^2}2,
\quad
s_y = x u_x u_y.
\end{aligned}
\end{equation}
Here $u_{xx} + u_{yy} + u_x/x = 0$ is the cylindrical Laplace equation.
As is well known, its solutions
correspond to axisymmetric solutions of the three-dimensional
Laplace equation.
Consequently, the admissible dust-vacuum boundaries, which are 
represented by the levels of $u$ by Theorem~\ref{thm:main},
correspond to equipotential surfaces of axisymmetric electrostatic 
potentials in dimension three.
Thus, the problem of finding van Stockum--Papapetrou dust clouds
of a prescribed shape reduces to that of
finding electrostatic fields with a prescribed equipotential surface,
which is a classical boundary problem in electrostatics. 

Since $u$ is a harmonic function, it is either constant or unbounded, 
and the boundary curve 
$u = u_0 =$ const is nonempty and regular
for the continuum of values of~$u_0$ in the interval 
$(\lim\inf u, \lim\sup u) = (-\infty,\infty)$. 
Observe that $u =$ const if and only if $w = $ const, 
in which case the Papapetrou 
metric is flat and has no van Stockum companion.

Choosing two different values for $u_0$,
we obtain a dust layer sandwiched between two Papapetrou vacua.
Moreover, disconnected as well as self-intersecting 
boundaries can occur.
Examples are presented in Section~\ref{sect:ex}. 
In the same vein, recent work~\cite{AE-DPS-2013} on the topology 
of level sets of harmonic functions in three dimensions reveals 
a rich topology of possible boundaries.

We end this section with some elementary facts related to the
electrostatic picture.

\begin{proposition}
Raising indices with $g^{\rm v}_{ij}$, we have
\[
u^{,i} u_{,i} = \rho^{\rm d},
\quad
f^{,i} u_{,i} = 0.
\]
\end{proposition}


\begin{proof}
By straightforward computation, using eq.~\eqref{eq:Papa:uws}.
\end{proof}

Since $u$ and $f$ depend only on the orbit space coordinates, and the
orbit space and the cut plane are locally conformally diffeomorphic, 
we have the following corollary.

\begin{corollary}
\label{cor:u f ortho}
Level sets of functions $u$ and $f$ intersect orthogonally in both the 
orbit space and the cut plane~$\RR^2$.
\end{corollary}

\begin{corollary}
In the electrostatic picture, the level sets of $f$ coincide with 
the electric field lines corresponding to the potential $u$.
\end{corollary}

\section{Relation to static Weyl vacuum metrics}
\label{sect:E-NK}

The following proposition can be easily verified by
straightforward computation.

\begin{proposition} 
\label{prop:Weyl}
The   dust and vacuum metrics~\eqref{eq:dust:ds:E} 
and~\eqref{eq:papa:ds:NK} are, respectively, the Ehlers and 
the Neugebauer--Kramer transform of the static Weyl vacuum metric
\begin{equation}
\label{eq:static}
\e^{u + p}(\d x^2 + \d y^2) + x^2 \e^u \d \phi^2 - \e^{-u} \d t^2.
\end{equation}
\end{proposition} 

For the Neugebauer--Kramer transform of a static metric 
see~\cite[\S~4.1]{DK-GN-1968}, for the Ehlers transform of a 
static metric see~\cite[Theorem~21.1]{Book} or the original 
paper~\cite{JE-1959}.

Thus, the companion correspondence 
can be decomposed as follows:
\[
\unitlength = 0.95mm
\begin{picture}(140,22)(0,25)
\put(0,25){\framebox(20,20)[c]{}}
\put(0,20){\makebox(20,35)[c]{\small Papapetrou}}
\put(0,20){\makebox(20,25)[c]{\small vacuum}}
\thicklines
\put(40,35){\vector(-1,0){18}}
\put(40,35){\vector(1,0){18}}
\put(40,37){\makebox(0,0)[c]{\footnotesize Neugebauer--Kramer}}
\thinlines
\put(60,25){\framebox(20,20)[c]{}}
\put(60,20){\makebox(20,40.5)[c]{\small static}}
\put(60,20){\makebox(20,30.5)[c]{\small Weyl}}
\put(60,20){\makebox(20,20.5)[c]{\small vacuum}}
\thicklines
\put(100,35){\vector(-1,0){18}}
\put(100,35){\vector(1,0){18}}
\put(100,37){\makebox(0,0)[c]{\footnotesize Ehlers}}
\thinlines
\put(120,25){\framebox(20,20)[c]{}}
\put(120,20){\makebox(20,40.5)[c]{\small van}}
\put(120,20){\makebox(20,30.5)[c]{\small Stockum}}
\put(120,20){\makebox(20,20.5)[c]{\small dust}}
\end{picture}
\]

The last proposition opens a way to reuse static Weyl 
metrics and obtain dust clouds of a required shape $u =$ const.
The coefficient 
$\e^{-u}$ at $\d t^2$ is the relativistic analogue of the Newtonian 
gravitational potential
(possibly somewhat distorted~\cite{L-O-1998}).
Otherwise said, boundaries of van Stockum dust clouds in 
Papapetrou vacuum correspond to Newtonian equipotential surfaces
in static Weyl space-times, even unphysical ones
(permitting negative masses). 

Essentially, all static Weyl vacua are known and many particular 
cases have been studied
(e.g., \cite{JBG-JP-2009,L-O-1998,SS-1999,OS-2016,Book} 
and references therein). 
According to Proposition~\ref{prop:Weyl}, each yields a matching 
van Stockum--Papapetrou pair along with an explicit expression 
for the field variables~$u$ and~$p$, sufficient to compute
the dust density and admissible borders.
By contrast, computing the coefficients $h_{kl}$ requires 
a closed-form representation for $f$, which is not always available, 
because $f$ is determined by a path integral
which is not known analytically in many cases 
(see Examples~\ref{ex:BW} and~\ref{ex:AGP}), although 
its levels can be inferred from the levels of $u$ by orthogonality 
(Corollary~\ref{cor:u f ortho}).
Yet all invariant quantities can be computed without
explicit knowledge of~$f$, including the Petrov type~\cite[Ch.~4]{Book} 
and the curvature invariants~\cite[Ch.~9]{Book},
as well as the first-order invariants~\cite{MM-OS-2008}.
Obviously from~\eqref{eq:Papa:uws}, the derivatives of $f$ and $p$ 
are expressible via derivatives of $u$, although $f$ itself is not.
This reflects the fact that $f \mapsto f + c$ corresponds to a 
coordinate transformation.

\section{Examples}
\label{sect:ex}

\begin{example} \rm
Table~\ref{tab:I} provides information about four cases when the 
dust part or the vacuum part or both have been studied earlier.

\begin{table}[ht]
\newdimen\pb
\pb = 4.75em
\def\pbch#1{\parbox{\pb}{\raggedright\footnotesize #1}}

\begin{center}
\def\D#1{$\displaystyle#1$}
\def\tstrut{\vrule height 4ex depth 2.4ex width 0ex}
\hglue-1em
\begin{tabular}{c|@{\,\tstrut}|c|c|c|c}
\pbch{electro\-static analogue} 
 & $u$
 & \D{f = w}
 & \D{p = s}
 & \D{\rho}
\\\hline
\noalign{\vskip-1pt}
\hline
\pbch{point charge} 
 & \D{\frac{2}{R}}
 & \D{\frac{2 y}{R}} 
 & \D{\llap{$-$}\frac{x^2}{(x^2 + y^2)^2}} 
 & \D{\frac{4}{(x^2 + y^2)^2\,\e^p}}
\\\hline
\pbch{point dipole}
 & \D{\llap{$-$}\frac{2 y}{R^3}} 
 & \D{\frac{2 x^2}{R^3}} 
 & \D{\frac{x^2 (x^2 - 8 y^2)}{2 (x^2 + y^2)^4}}
 & \D{\frac{4 (x^2 + 4 y^2)}{(x^2 + y^2)^4\,\e^p}}
\\\hline
\pbch{infinite plate}
 & \D{2 y}
 & \D{x^2}
 & \D{\llap{$-$}x^2}
 & \D{4 \e^{x^2}}
\\\hline
\pbch{finite rod}
 & \D{\ln\left(\frac{y_+ + R_{\rm II}^+}{y_- + R_{\rm II}^-} \right)}
 & \D{R_{\rm II}^+ - R_{\rm II}^-}
 & \D{\ln\left(1 + \frac{N_{\rm II}}{R_{\rm II}^+ R_{\rm II}^-}\right)}
 & \D{\frac{(N_{\rm II} - R_{\rm II}^+ R_{\rm II}^-)^2}{2 b^2 x^4}}
\end{tabular}

$$
R = \sqrt{x^2 + y^2}, 
\quad
y_\pm = y \pm b, 
\quad
R_{\rm II}^{\pm} = \sqrt{\smash[b]{x^2 + y_\pm^2}},
\quad
N_{\rm II} = x^2 + y^2 - b^2.
$$
\end{center}
\caption{Examples of matching  
van Stockum dust to Papapetrou vacuum}
\label{tab:I}
\end{table}

From left to right, the columns harbour  
the axisymmetric charge geometry (see Section~\ref{sect:el}); 
the potential $u$ the levels of which determine the 
boundaries according to Theorem~\ref{thm:main}; 
the field variables $w$ and $s$ as computed
from system~\eqref{eq:Papa:uws} and equal, respectively, to
$f$ and $p$ of the companion metric; and the dust density 
\begin{equation}
\label{eq:mu:u}
\rho = \frac{u_x^2 + u_y^2}{\e^p}
\end{equation}
as obtained from~\eqref{eq:dust:rho} and~\eqref{eq:Papa:uws}.
The dust and vacuum metrics are given by 
equations~\eqref{eq:dust:ds:E} and~\eqref{eq:papa:ds:NK}, 
respectively. 
The dust's four-velocity is $\delta^i_4$
(comoving coordinates).
Unessential constant parameters are suppressed.

Let us comment on individual rows of Table~\ref{tab:I}.
For visualisations see the end of this section.

{\it 1. } 
The vacuum is the spinning metric due to Halilsoy~\cite{MH-1992}, 
which, contrary to its name, is not rotating~\cite{N-M-A-B-2020}.
For this reason, the matching surfaces $R =$ const are not 
interpretable as spheres. 
The static seed (see Section~\ref{sect:E-NK}) is the Chazy--Curzon 
metric.

{\it 2.} {\it The Bonnor--Bonnor cloud. }
The dust part is the well-known Bonnor metric~\cite{WBB-1977}.
The vacuum is due to Bonnor~\cite{WBB-2005} as well, but the finding
that both metrics match is new.
All boundaries traverse the singularity located at the centre.

{\it 3. }
\label{rem:wall}
The dust part is the Lanczos~\cite{KL-1924} metric 
(the ``cylindrical world''), 
which has been rediscovered by van Stockum~\cite{WJvS-1937} and
matched to a cylindrically symmetric Lewis~\cite{TL-1932}
vacuum metric 
along $x = $ const, yielding an infinite dust cylinder.
Our matching along $y =$ const breaks the cylindrical symmetry
and results in a thick wall of rigidly rotating dust extending 
to infinity 
(ignoring what happens when the density becomes too high).
This confirms that one and the same dust solution can match to 
different vacua along different boundaries 
(which is not true for perfect fluids, where the boundary 
occurs at the zero level of pressure).

{\it 4.} {\it The Zsigrai cloud. }
The Zsigrai~\cite{JZ-2003} metric results from gluing the 
Luk\'acs--Newman--Sparling--Winicour dust metric~\cite{L-N-S-W-1983}
to the vacuum NUT metric~\cite{N-T-U-1963}.
Matching is along isodensity surfaces since 
\[
\rho = \frac2{b^2}\sinh^4 \frac{u}{2}.
\]
in this case.
\end{example}

In the following two examples we reuse known toroidal static 
Weyl metrics to produce axisymmetric van Stockum--Papapetrou dust 
clouds (they satisfy the regular axis condition~\cite[\S~19.1]{Book}).
Although closed-form representations for $f$ are not available, we are 
able to understand the topologies the clouds can have.
All clouds are named after the static vacuum seed.


\begin{example} 
\label{ex:BW}
\rm {\it Rotating Bach--Weyl cloud. }
The seed is the static Bach--Weyl 
solution~\cite{B-W-1922}.
In Weyl coordinates, we can write (Semer\'ak~\cite[III.B]{OS-2016})
\[
\begin{aligned}
u &= \frac{4 m K(\Omega)}{\sqrt{(x + a)^2 + y^2}},
\qquad
\Omega = 2 \sqrt{\frac{a x}{(x + a)^2 + y^2}},
\\
p &= -\frac{m}{a^2} 
     \biggl(
       \frac{x^2 + y^2 + 3 a^2}{(x + a)^2 + y^2}\,K(\Omega)^2
       - 2\,K(\Omega)\,E(\Omega)
       + \frac{x^2 + y^2 - a^2}{(x - a)^2 + y^2}\,E(\Omega)^2
     \biggr),
\end{aligned}
\]
where $K,E$ denote the complete elliptic functions.  

Originally obtained as the gravitational field of a static ring,
it can be also obtained as an 
invariant solution with respect to the Lie symmetry
$2 x y u_x + (a^2 - x^2 + y^2) u_y + y u$, as can be easily checked.
The density~\eqref{eq:mu:u} is   
\[
\rho = \frac{4 m^2}{x^2\,\e^p} \biggl(
  \frac{K(\Omega)^2}{(x + a)^2 + y^2}
   - 2\,(a^2 - x^2 + y^2)
     \frac{K(\Omega)}{(x + a)^2 + y^2}\,
     \frac{E(\Omega)}{(x - a)^2 + y^2}
  + \frac{E(\Omega)^2}{(x - a)^2 + y^2}
\biggr).
\]
Boundaries can have one or two components, with a transient 
eightlike boundary passing through the saddle point $2 m \pi/a$.
\end{example}

\begin{example} 
\label{ex:AGP}
\rm {\it Rotating Appell--Gleiser--Pullin cloud. }
The seed is the static axisymmetric solution introduced by
Gleiser and Pullin~\cite{G-P-1989}, which incorporates
Appell's~\cite{PA-1887} harmonic function possessing a
circular singularity.
In Weyl's coordinates, the solution is determined by
\[
\begin{aligned} 
u &= \sqrt{\frac{\sqrt{4 a^2 y^2 + (x^2 + y^2 - a^2)^2} + x^2 + y^2 - a^2}
 {4 a^2 y^2 + (x^2 + y^2 - a^2)^2}},
\\
p &= -x^2 \frac{(x^2 + y^2 - 2 a y - a^2) (x^2 + y^2 + 2 a y - a^2)}
  {4 \bigl((x + a)^2 + y^2 \bigr)^2 \bigl((x - a)^2 + y^2 \bigr)^2}
 - \frac{x^2 + y^2 + a^2}{8 a^2\sqrt{(x - a)^2 + y^2}\sqrt{(x + a)^2 + y^2}}
\end{aligned}
\]
(Semer\'ak~\cite[III.C]{OS-2016}).
The density~\eqref{eq:mu:u} is
\[
\begin{aligned}
&\rho = \frac{a^2 - x^2 - y^2 + \sqrt{(a + x)^2 + y^2} \sqrt{(a - x)^2 + y^2}}
     {2 a^2 y^2 \bigl((a + x)^2 + y^2 \bigr)^2 \bigl((a - x)^2 + y^2\bigr)^2 \e^p}
\\&\quad\times\Bigl(
  \bigl((x^2 + y^2) (x^2 + y^2 - a^2) - 2 a^2 y^2
\bigr)
  \sqrt{(a + x)^2 + y^2} \sqrt{(a - x)^2 + y^2}
\\&\quad\quad\quad + (x^2 + y^2) (x^2 + y^2 - a^2)^2 + 4 a^4 y^2
  \Bigr).
\end{aligned}
\]
The function $u$ has two saddle points $(0,a)$, $(0,-a)$,
where it assumes the value $u(0,\pm a) = \sqrt2/2a$.
All level curves (boundaries) have two components, except the transient 
one, which is self-intersecting.
Curves entering the singularity have cusps there.
\end{example}

For visualisation see Appendix A.
A rich variety of admissible shapes can be seen.
We already noted the possibility of hollow Bonnor--Bonnor clouds.
The Bach--Weyl clouds can be also toroids, 
ovaloids containing an ovaloidal or toroidal hole,
toroids containing a toroidal hole.
Finally, the Appell--Gleiser--Pullin clouds can be
also two disjoint (nested) hollow ovaloids, corresponding to two 
different two-component level sets of $u$.

\section{Discussion}
\label{sect:disscussion}

What is really surprising is that the problem of matching 
van Stockum dust to vacuum has been waiting for solution so long,
considering the simplicity of the answer and the demand for it 
\cite{WBB-1982,NG-2009,SV-2007,Z-A-T-2007}. Not only are the Papapetrou 
and van Stockum metrics widely known, they also turn out to 
be rather natural candidates for matching.
A hint from physics is that asymptotically flat rotating Papapetrou 
metrics require a zero-mass source, see~\cite[\S~2.5]{JNI-book} 
or \cite[\S~20.3]{Book}, while the overall mass of the van Stockum dust 
is zero, since a negative mass singularity balances the positive mass of 
the dust (Bonnor~\cite{WBB-1977}, Bratek et al.~\cite{B-J-K-2007}).

That said, we must also admit that rotating Papapetrou vacua 
have no known physical interpretation other than being a zero-mass 
limit \cite{WBB-2005,AS-1971}, while serious doubts persist about 
whether van Stockum dust can exist in nature 
\cite{AC-1978,JF-1987,NG-2009,HP-2010,DRR-2015,Z-A-T-2007}. 
It is, however, no less true that negative masses have been 
admitted as constituents of relativistic models repeatedly during the 
last decades, suggesting that van Stockum metrics can avoid the fate 
of being unphysical.  
As a case in point, Ilyas et al.~\cite{I-Y-M-B-2017} proposed a 
measurement to identify possible occurrence of the Bonnor 
dust~\cite{WBB-1977} in a galaxy centre.
Anyhow, investigation of wider classes of dust-vacuum 
$C^1$-metrics, of which our Papapetrou--van Stockum class would be 
a limiting case, is under way.

\section*{Acknowledgements}

The author is grateful to R.~Unge and J.~Novotn\'y for valuable 
discussions and advice. 
This research received support from M\v{S}MT under RVO 47813059.

%

\catcode`/=12\catcode`_=8

\appendix

\section{Visualisation}

Figure~A.1 shows three examples in axial section.

\begin{figure}[h]
\begin{center}
\includegraphics[scale=1.1]{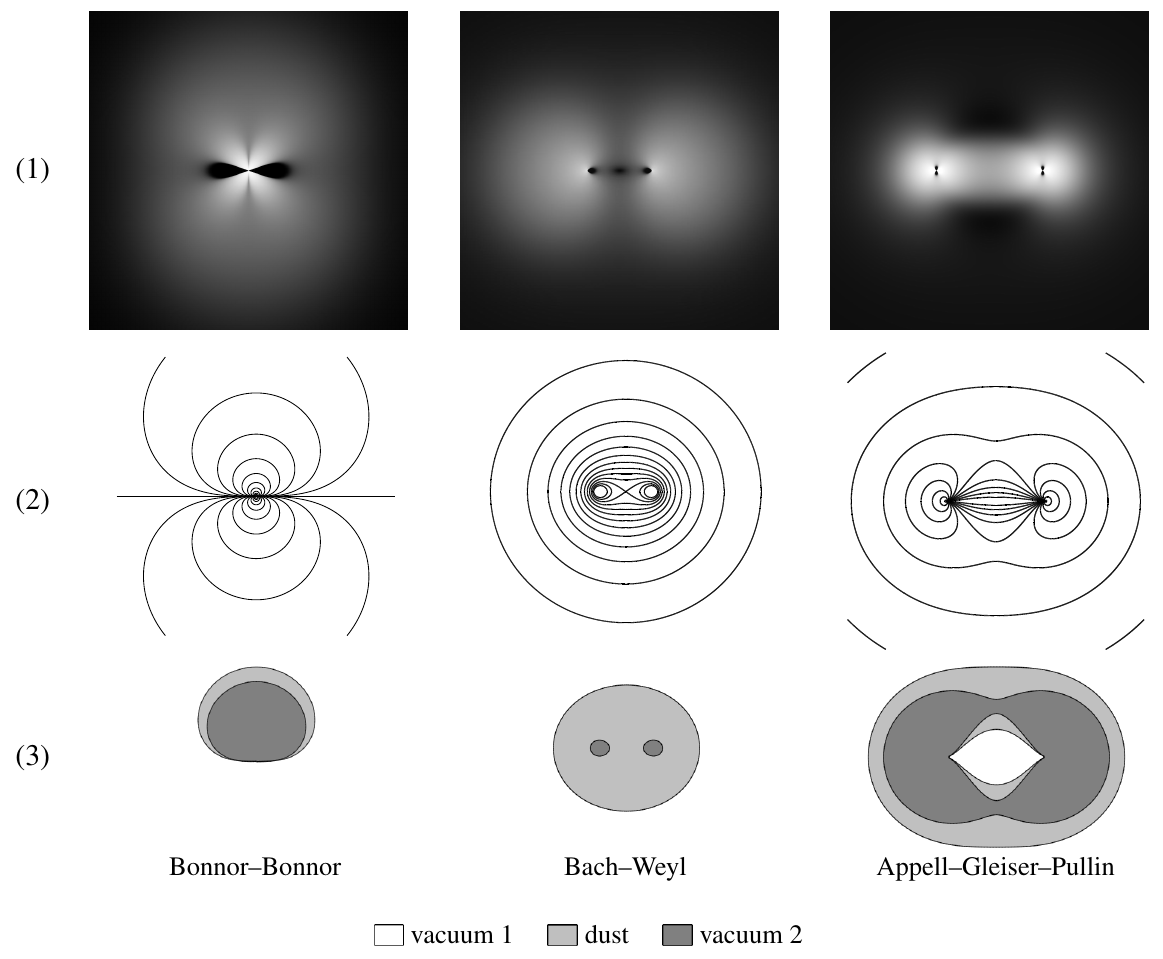}
\end{center}
\caption{\ Axisymmetric dust clouds examples in Weyl's coordinates.
(1)~Densities (increasing from dark to light).
(2)~Admissible boundaries. 
(3)~Example clouds.} 
\label{fig:clouds}
\end{figure}

The eightlike voids (to be viewed zoomed in) 
reveal directional singularities, which are
artifacts of the Weyl coordinates (Taylor~\cite{JPWT-2005}).

\end{document}